\pdfoutput=1

\RequirePackage[l2tabu,orthodox]{nag}

\documentclass[11pt,letterpaper,reqno,oneside]{article}

\DeclareTextFontCommand{\emph}{\slshape}

\makeatletter
\renewcommand{\paragraph}{%
	\@startsection{paragraph}{4}%
	{\z@}{1.75ex \@plus 1ex \@minus .2ex}{-0.7em}%
	{\normalfont\normalsize\bfseries}%
}
\makeatother

\usepackage[T1]{fontenc}
\usepackage{lmodern}
\usepackage[utf8]{inputenc}

\usepackage[american,german,british]{babel}
\usepackage[activate={true,nocompatibility}]{microtype}

\usepackage{csquotes}

\usepackage[backend=biber,autolang=other,style=apa,maxcitenames=5,sortcites=false]{biblatex}
\DeclareLanguageMapping{british}{british-apa}
\DeclareLanguageMapping{american}{american-apa}
\usepackage{amsmath}
\usepackage{amssymb}
\usepackage{amsfonts}
\usepackage{amsthm}
\usepackage{mathtools}
\usepackage{stmaryrd}
\let\originalleft\left
\let\originalright\right
\renewcommand{\left}{\mathopen{}\mathclose\bgroup\originalleft}
\renewcommand{\right}{\aftergroup\egroup\originalright}

\usepackage{graphicx}
\usepackage{tikz,pgfplots}
\pgfplotsset{compat=1.10}
\usetikzlibrary{shapes.multipart}
\usetikzlibrary{decorations.markings}
\usepackage[title]{appendix}

\usepackage{datetime}
\newdateformat{datestyle}{ {\THEDAY} \monthname[\THEMONTH] \THEYEAR }

\usepackage{enumerate}
\usepackage{enumitem}
\setlist[enumerate,1]{label=(\arabic*)}
\setlist[itemize,1]{label=--}
\setlist[itemize,2]{label=--}
\setlist[itemize,3]{label=--}
\setlist[itemize,4]{label=--}

\usepackage{xcolor}
\usepackage{verbatim}
\usepackage{gensymb}
\usepackage{rotating}
\usepackage{subcaption}
\usepackage{floatpag}
\floatpagestyle{plain}
\usepackage{marvosym}
\usepackage{hyphenat}
\theoremstyle{definition}
\newtheorem{proposition}{Proposition}%[section]
\newtheorem{lemma}{Lemma}%[section]
\newtheorem{corollary}{Corollary}%[section]
\newtheorem{remark}{Remark}%[section]
\newtheorem{observation}{Observation}%[section]
\newtheorem{example}{Example}%[section]
\newtheorem{definition}{Definition}%[section]
\newtheorem{axiom}{Axiom}%[section]

\newtheoremstyle{named}
	{\topsep}					% ABOVESPACE
	{\topsep}					% BELOWSPACE
	{}							% BODYFONT
	{0pt}						% INDENT (empty value is the same as 0pt)
	{\bfseries}					% HEADFONT
	{}							% HEADPUNCT
	{5pt plus 1pt minus 1pt}	% HEADSPACE
	{\thmnote{#3}}				% CUSTOM-HEAD-SPEC
\theoremstyle{named}
\newtheorem{namedthm}{}

\usepackage{xpatch}
\xpatchcmd{\proof}{\itshape}{\proofheadfont}{}{}
\newcommand{\proofheadfont}{\slshape}

\usepackage{hyperref}
\hypersetup{pdfborder=0 0 0}

\usepackage[nameinlink]{cleveref}
\crefname{page}{p.}{pp.}
\crefname{equation}{equation}{equations}
\crefname{section}{section}{sections}
\crefname{subsection}{section}{sections}
\crefname{subsubsection}{section}{sections}
\crefname{appsec}{appendix}{appendices}
\crefname{supplsec}{supplemental appendix}{supplemental appendices}
\crefname{footnote}{footnote}{footnotes}
\crefname{figure}{figure}{figures}
\crefname{table}{table}{tables}
\crefname{theorem}{theorem}{theorems}
\crefname{proposition}{proposition}{propositions}
\crefname{lemma}{lemma}{lemmata}
\crefname{corollary}{corollary}{corollaries}
\crefname{remark}{remark}{remarks}
\crefname{observation}{observation}{observations}
\crefname{example}{example}{examples}
\crefname{fact}{fact}{facts}
\crefname{definition}{definition}{definitions}
\crefname{assumption}{assumption}{assumptions}
\crefname{exercise}{exercise}{exercises}
\crefname{notation}{notation}{notation}
\crefname{claim}{claim}{claims}
\crefname{conjecture}{conjecture}{conjectures}
\crefname{axiom}{axiom}{axioms}

\newcommand{\dd}{\mathrm{d}}

\newcommand{\E}{\mathbf{E}}

\newcommand{\R}{\mathbf{R}}

\newcommand{\N}{\mathbf{N}}

\DeclarePairedDelimiter\abs{\lvert}{\rvert}

\newcommand*{\xslant}[2][76]{%
	\begingroup
	\sbox0{#2}%
	\pgfmathsetlengthmacro\wdslant{\the\wd0 + cos(#1)*\the\wd0}%
	\leavevmode
	\hbox to \wdslant{\hss
		\tikz[
			baseline=(X.base),
			inner sep=0pt,
			transform canvas={xslant=cos(#1)},
		] \node (X) {\usebox0};%
		\hss
		\vrule width 0pt height\ht0 depth\dp0 %
	}%
	\endgroup
}
\makeatletter
\newcommand*{\xslantmath}{}
\def\xslantmath#1#{%
	\@xslantmath{#1}%
}
\newcommand*{\@xslantmath}[2]{%
	\ensuremath{%
		\mathpalette{\@@xslantmath{#1}}{#2}%
	}%
}
\newcommand*{\@@xslantmath}[3]{%
	\xslant#1{$#2#3\m@th$}%
}
\makeatother

\makeatletter
\def\namedlabel#1#2{\begingroup
	#2%
	\def\@currentlabel{#2}%
	\phantomsection\label{#1}\endgroup
}
\makeatother

\makeatletter
\let\save@mathaccent\mathaccent
\newcommand*\if@single[3]{%
	\setbox0\hbox{${\mathaccent"0362{#1}}^H$}%
	\setbox2\hbox{${\mathaccent"0362{\kern0pt#1}}^H$}%
	\ifdim\ht0=\ht2 #3\else #2\fi
	}
\newcommand*\rel@kern[1]{\kern#1\dimexpr\macc@kerna}
\newcommand*\widebar[1]{\@ifnextchar^{{\wide@bar{#1}{0}}}{\wide@bar{#1}{1}}}
\newcommand*\wide@bar[2]{\if@single{#1}{\wide@bar@{#1}{#2}{1}}{\wide@bar@{#1}{#2}{2}}}
\newcommand*\wide@bar@[3]{%
	\begingroup
	\def\mathaccent##1##2{%
	  \let\mathaccent\save@mathaccent
	  \if#32 \let\macc@nucleus\first@char \fi
	  \setbox\z@\hbox{$\macc@style{\macc@nucleus}_{}$}%
	  \setbox\tw@\hbox{$\macc@style{\macc@nucleus}{}_{}$}%
	  \dimen@\wd\tw@
	  \advance\dimen@-\wd\z@
	  \divide\dimen@ 3
	  \@tempdima\wd\tw@
	  \advance\@tempdima-\scriptspace
	  \divide\@tempdima 10
	  \advance\dimen@-\@tempdima
	  \ifdim\dimen@>\z@ \dimen@0pt\fi
	  \rel@kern{0.6}\kern-\dimen@
	  \if#31
	    \overline{\rel@kern{-0.6}\kern\dimen@\macc@nucleus\rel@kern{0.4}\kern\dimen@}%
	    \advance\dimen@0.4\dimexpr\macc@kerna
	    \let\final@kern#2%
	    \ifdim\dimen@<\z@ \let\final@kern1\fi
	    \if\final@kern1 \kern-\dimen@\fi
	  \else
	    \overline{\rel@kern{-0.6}\kern\dimen@#1}%
	  \fi
	}%
	\macc@depth\@ne
	\let\math@bgroup\@empty \let\math@egroup\macc@set@skewchar
	\mathsurround\z@ \frozen@everymath{\mathgroup\macc@group\relax}%
	\macc@set@skewchar\relax
	\let\mathaccentV\macc@nested@a
	\if#31
	  \macc@nested@a\relax111{#1}%
	\else
	  \def\gobble@till@marker##1\endmarker{}%
	  \futurelet\first@char\gobble@till@marker#1\endmarker
	  \ifcat\noexpand\first@char A\else
	    \def\first@char{}%
	  \fi
	  \macc@nested@a\relax111{\first@char}%
	\fi
	\endgroup
}
\makeatother

\makeatletter
	\newcommand{\hyperdest}[1]{\Hy@raisedlink{\hypertarget{#1}{}}}
\makeatother

\title{\scshape Optimism, overconfidence,\\and moral hazard%
\thanks{This paper grew out of conversations with Jorge Padilla and Joe Perkins, whom I thank for their insights. I am grateful for comments from three anonymous referees, Johannes Abeler, Da\u{g}han Carlos Akkar, Nemanja Antić, Ala Avoyan, Jean Baccelli, Aislinn Bohren, Roberto Corrao, Gregorio Curello, Eddie Dekel, Matteo Escudé, Bruno Furtado, George Georgiadis, Duarte Gonçalves, Tai-Wei Hu, Alex Jakobsen, Ian Jewitt, Jan Knoepfle, Meg Meyer, Salvatore Piccolo, Mauricio Ribeiro, Evgenii Safonov, Todd Sarver, Lorenzo Stanca, Quitzé Valenzuela-Stookey, Mu Zhang and audiences at Compass Lexecon and the 2nd Southeast Theory Festival. This research was supported by Compass Lexecon. An earlier version of this paper bore the title `Optimism and overconfidence'.}}

\author{Ludvig Sinander\\University of Oxford}

\date{24 May 2024}

\makeatletter
	\AtBeginDocument{ \hypersetup{
		pdftitle = {Optimism, overconfidence, and moral hazard},
		pdfauthor = {Ludvig Sinander}
		} }
\makeatother

\begin{document}

\maketitle

%%%%%%%%%%%%%%%%%%%%%%
%%%%%%%%%%%%%%%%%%%%%%
\begin{abstract}
	I revisit the standard moral-hazard model, in which an agent's preference over contracts is rooted in costly effort choice. I characterise the behavioural content of the model in terms of empirically testable axioms, and show that the model's parameters are identified. I propose general behavioural definitions of relative (over)confidence and optimism, and characterise these in terms of the parameters of the moral-hazard model. My formal results are rooted in a simple but powerful insight: that the moral-hazard model is closely related to the well-known `variational' model of choice under uncertainty.

	\vspace{0.8em}


\end{abstract}
%%%%%%%%%%%%%%%%%%%%%%
%%%%%%%%%%%%%%%%%%%%%%

%%%%%%%%%%%%%%%%%%%%%%
%%%%%%%%%%%%%%%%%%%%%%
\section{Introduction}
\label{sec:intro}
%%%%%%%%%%%%%%%%%%%%%%
%%%%%%%%%%%%%%%%%%%%%%

A vast literature in psychology and economics documents the prevalence of wrong beliefs.%
	\footnote{See the survey by \textcite[][section 3.1]{Dellavigna2009}, for example.}
In this paper, I seek to define, distinguish and characterise two oft-conflated senses in which beliefs can be `wrong': \emph{overconfidence} and \emph{optimism.} An agent is \emph{overconfident} if she overestimates her ability to influence (the distribution of) a payoff-relevant outcome, which I shall call `output'. By contrast, she is \emph{optimistic} if her expectation of the distribution of output is unrealistically high, in the sense of first-order stochastic dominance. In simple parametric models, these two concepts typically admit natural definitions in terms of the parameters.%
	\footnote{See \textcite{Spinnewijn2015}, for example.}
In this paper, I provide general, model-independent definitions of overconfidence and optimism in terms of choice behaviour.

I explicate the content of my definitions by characterising their meaning in the standard moral-hazard model. This model is the natural benchmark because it represents the canonical formalism in economics for studying how an agent may influence the distribution of an observable outcome (`output'). The question of how overconfident and optimistic behaviour manifest in the moral-hazard model raises a more basic question: how does this model relate to choice behaviour? I answer this question by characterising what restrictions the moral-hazard model imposes on choice data, and delineating the extent to which its parameters may be recovered from such data.

The environment consists of a finite set $S$ of possible output levels and a convex set $\Pi \subseteq \R$ of possible monetary rewards. A \emph{contract} $w : S \to \Delta(\Pi)$ specifies the agent's (possibly random) remuneration as a function of output. In the \emph{moral-hazard model,} the agent influences the distribution $P_e \in \Delta(S)$ of output by choosing `effort' $e \in E$, at a cost $C(e) \geq 0$. She chooses effort optimally, randomising if desired (by selecting an effort distribution $\mu \in \Delta(E)$). The agent's valuation of a contract $w : S \to \Delta(\Pi)$ is therefore
\begin{equation*}
	\sup_{\mu \in \Delta(E)} \int_E \left[
	-C(e)
	+ \sum_{s \in S} u(w(s)) P_e(s)
	\right]
	\mu(\dd e) ,
\end{equation*}
where the utility function $u : \Pi \to \R$ describes her risk attitude.%
	\footnote{For a random remuneration $x \in \Delta(\Pi)$, $u(x)$ denotes the \emph{expected} utility $\int_\Pi u(\pi) x(\dd \pi)$.}

I first ask what the moral-hazard model's empirical content is: what testable restrictions does this model impose on contract-choice data? In other words, which preference relations $\succeq$ on the set of all contracts admit a moral-hazard representation? I give the answer in terms of six axioms, which together exhaust the testable implications of the moral-hazard model.

I then show that the moral-hazard model's parameters are to a significant extent identified off contract-choice data. In particular, the agent's utility function $u$ is identified (up to positive affine transformations), and so is her minimum cost of producing any given output distribution $p \in \Delta(S)$: that is, the quantity
\begin{equation}
	c(p) \coloneqq
	\inf_{\substack{\mu \in \Delta(E) :\\ \int_E P_e \mu(\dd e) = p}} \int_E C(e) \mu(\dd e) 
	\qquad \text{for each $p \in \Delta(S)$.}
	\tag{$\star$}
	\label{eq:c}
\end{equation}

I next define relative confidence, motivated by the idea that a confident agent is one who believes that she can significantly influence output. Specifically, I call one preference $\succeq$ \emph{more confident than} another preference $\succeq'$ if whenever $\succeq'$ chooses a contract $w$ over a \emph{constant} contract $x$ (one under which pay does not vary with output), $\succeq$ also chooses $w$ over $x$. I show that for moral-hazard preferences, confidence shifts are equivalent to \emph{vertical shifts} of the output-distribution cost $c$ defined in \eqref{eq:c}: greater confidence means precisely a pointwise lower $c$ function (and an unchanged utility function $u$).

Finally, I define relative optimism, building on the idea that an optimistic agent is one who believes that output is likely to be high. In the simplest case, for two preferences $\succeq$ and $\succeq'$ that both have a risk-neutral expected-utility attitude to risk, I call $\succeq$ \emph{more optimistic than} $\succeq'$ if and only if whenever $\succeq'$ chooses a contract $w$ over another contract $w'$ that is \emph{less steep,} in the sense that $s \mapsto \E(w(s))-\E(w'(s))$ is increasing, $\succeq$ also chooses $w$ over $w'$. The general definition is similar, but adjusts for risk attitude. I show that in the moral-hazard model, optimism shifts correspond to \emph{horizontal shifts} of the output-distribution cost $c$ in \eqref{eq:c}: greater optimism means a shift of the $c$ function toward the first-order stochastically higher distributions (and no shift of the utility function $u$).

My definitions of relative confidence and optimism are purely behavioural, making no reference to objective facts about the agent's \emph{actual} ability to influence output. \emph{If} an objective description of the agent's influence is given, then I may additionally define what it means for a preference to be (absolutely) `overconfident' or `optimistic': namely, `more confident than objectively warranted' and `more optimistic than objectively warranted'.

The formal results in this paper are rooted in a simple but powerful insight: that the moral-hazard model is closely related to the `variational' model of choice under uncertainty, which was introduced and axiomatised by \textcite{MaccheroniMarinacciRustichini2006}. This close relationship allows me to leverage known results about the variational model to obtain my four main results about the moral-hazard model (axiomatisation, identification, characterisation of `more confident than', and characterisation of `more optimistic than'). Specifically, I apply results from \textcite{MaccheroniMarinacciRustichini2006} and \textcite{DziewulskiQuah2024}.

%%%%%%%%%%%%%%%%%%%%%%%%%%%%%%%%%%%
\subsection{Related literature}
\label{sec:intro:lit}
%%%%%%%%%%%%%%%%%%%%%%%%%%%%%%%%%%%

The moral-hazard model is the focus of a large literature, beginning with \textcite{Mirrlees1999,Holmstrom1979}.%
	\footnote{For an overview, see \textcite{Holmstrom2017,Georgiadis2022}.}
	% \footnote{Precedents include \textcite{Wilson1969,SpenceZeckhauser1971,Ross1972,Stiglitz1975}.}
I contribute to the `behavioural' strand of this literature, which examines, among other things, how greater confidence or optimism on the part of the agent shapes optimal contracting.%
	\footnote{\label{footnote:mh_opt}For example, \textcite{Dellarosa2011,GervaisHeatonOdean2011,Spinnewijn2015}. See \textcite[section 3.4]{Koszegi2014} for a survey.}
This literature takes the moral-hazard model as given and studies its implications for optimal contracts; in this paper, I provide behavioural foundations for the model itself, as well as for overconfidence and optimism.

My characterisation of the testable implications of the moral-hazard model is designed to respect the model's two defining features: that effort is \emph{costly,} and that it is \emph{unobservable.} This distinguishes my analysis from previous work on the behavioural content of the moral-hazard model, which has either confined attention to the special case of costless effort \parencite{Dreze1961,Dreze1987}, or else assumed that effort is observable \parencite{Karni2006}.

I also contribute to the theoretical literature on behavioural economics, which studies models of individual decision-making that can accommodate various departures from `standard' economic behaviour that have been documented in the lab (or, sometimes, in the field). One strand of this literature examines `behavioural' models from the axiomatic perspective of decision theory;%
	\footnote{Examples include disappointment-aversion \parencite{Gul1991}, temptation \parencite{GulPesendorfer2001,DekelLipmanRustichini2009}, status-quo bias \parencite{MasatliogluOk2005,Ortoleva2010}, framing effects \parencite{SalantRubinstein2006,SalantRubinstein2008,AhnErgin2010}, heuristics \parencite{ManziniMariotti2007,ManziniMariotti2012}, regret-aversion \parencite{Sarver2008}, costly contemplation \parencite{ErginSarver2010}, and reference-dependence \parencite{OkOrtolevaRiella2015}.}
another strand is concerned with overconfidence and optimism.%
	\footnote{Topics include the impact of overconfidence or optimism on credit markets \parencite{ManovePadilla1999}, on competitive screening \parencite{SandroniSquintani2007}, on monopolistic screening \parencite{EliazSpiegler2008,Grubb2009}, on moral hazard (see \cref{footnote:mh_opt}), on political behaviour \parencite{OrtolevaSnowberg2015} and on learning \parencite{HeidhuesKoszegiStrack2018}. See also \textcite{BenoitDubra2011}.}
This paper marries these two strands. Relative to other work in that general spirit,%
	\footnote{E.g. \textcite{Wakker1990,EpsteinKopylov2007,ChateauneufEichbergerGrant2007,MooreHealy2008,DeanOrtoleva2017,DillenbergerPostlewaiteRozen2017}.}
I give new and simple definitions that distinguish overconfidence from optimism, and cash out their meaning in the moral-hazard model.

Finally, at a high level, this paper provides axiomatic behavioural foundations for a standard economic model (namely, the moral-hazard model). Other work in this spirit includes axiomatisations of Blackwell's (\citeyear{Blackwell1951,Blackwell1953}) model of choice informed by a signal,%
	\footnote{\cite{AzrieliLehrer2008}; \cite{DillenbergerEtal2014}; \cite{Lu2016}.}
of the canonical model of costly information acquisition,%
	\footnote{\cite{Vanzandt1996}; \cite{DeoliveiraEtal2017}; \cite{Ellis2018}.}
and of Kamenica and Gentzkow's (\citeyear{KamenicaGentzkow2011}) `Bayesian persuasion' model \parencite{Jakobsen2021,Jakobsen2024}.

%%%%%%%%%%%%%%%%%%%%%%%%%%%%%%%%%%%
\subsection{Roadmap}
\label{sec:intro:roadmap}
%%%%%%%%%%%%%%%%%%%%%%%%%%%%%%%%%%%

The rest of this paper is arranged as follows. In the next section, I describe the environment and the moral-hazard model. I introduce a canonical \emph{parsimonious} moral-hazard model in \cref{sec:parsim}. In \cref{sec:axioms}, I characterise the moral-hazard model in terms of six empirically testable axioms. I establish the model's identification properties in \cref{sec:ident}. Finally, in \cref{sec:conf_optim}, I propose behavioural definitions of `more confident than' and `more optimistic than', and characterise their meaning in the moral-hazard model.

%%%%%%%%%%%%%%%%%%%%%%
%%%%%%%%%%%%%%%%%%%%%%
\section{Setting}
\label{sec:model}
%%%%%%%%%%%%%%%%%%%%%%
%%%%%%%%%%%%%%%%%%%%%%

There is one agent, and a convex set $\Pi \subseteq \R$ of possible levels of monetary remuneration (`prizes'). The agent's pay $\pi \in \Pi$ can be made contingent on the realisation of a contractible signal. We follow convention by calling this signal `output'. The set $S$ of possible output realisations is assumed to be non-empty and, for simplicity, finite.

Let $\Delta(\Pi)$ be the set of all finite-support probability distributions over monetary prizes $\Pi$; we call each $x \in \Delta(\Pi)$ a \emph{random remuneration.} A \emph{contract} is a map $w : S \to \Delta(\Pi)$. The interpretation is that if realised output is $s \in S$, then the agent is paid a random amount drawn from the distribution $w(s) \in \Delta(\Pi)$. We write $W$ for the set of all contracts.

%%%%%%%%%%%%%%%%%%%%%%%%%%%%%%%%%%%
\subsection{Definitions and conventions}
\label{sec:model:prelim}
%%%%%%%%%%%%%%%%%%%%%%%%%%%%%%%%%%%

A contract $w \in W$ is \emph{constant} if and only if $w(s) = w(s')$ for all output levels $s,s' \in S$. As is standard, we identify each constant contract (an element of $W$) with the random remuneration at which it is constant (an element of $\Delta(\Pi)$). Similarly, we identify each degenerate random remuneration (an element of $\Delta(\Pi)$) with the prize at which it is degenerate (an element of $\Pi$).

Any (utility) function $u : \Pi \to \R$ defined on monetary prizes $\Pi$ extends naturally to a linear (expected-utility) function $u : \Delta(\Pi) \to \R$ defined on random remunerations, via
\begin{equation*}
	u(x) \coloneqq \int_\Pi u(\pi) x(\dd \pi)
	\quad \text{for every $x \in \Delta(\Pi)$.}
\end{equation*}
Throughout the paper, we hold fixed an arbitrary pair $\pi_0 < \pi_1$ of `reference prizes' in $\Pi$, and we call a (utility) function $u : \Pi \to \R$ such that $u(\pi_0) \neq u(\pi_1)$ \emph{normalised} if and only if $\{ u(\pi_0), u(\pi_1) \} = \{0,1\}$. This is merely an (arbitrary) choice of units in which to measure utility.

Let $\Delta(S)$ denote the set of all output distributions: that is, probability distributions on $S$. For any non-empty set $A$, call a function $\phi : A \to [-\infty,\infty]$ \emph{grounded} if and only if $\inf_{a \in A} \phi(a) = 0$. In the sequel, subsets of Euclidean spaces (such as the simplex $\Delta(S)$) are always equipped with the Borel $\sigma$-algebra, and `increasing' always means `weakly increasing'.

%%%%%%%%%%%%%%%%%%%%%%%%%%%%%%%%%%%
\subsection{The standard moral-hazard model}
\label{sec:model:MH}
%%%%%%%%%%%%%%%%%%%%%%%%%%%%%%%%%%%

In the moral-hazard model \parencite{Mirrlees1999,Holmstrom1979}, the agent chooses her `effort' from a non-empty, compact and convex subset $E$ of a Euclidean space; for example, $E = [0,\bar e]$ for some $\bar e \in (0,\infty)$. Each effort level $e \in E$ incurs a cost $C(e) \geq 0$ and produces a distribution $P_e \in \Delta(S)$ of output, where the cost function $C : E \to \R_+$ is grounded and lower semi-continuous, and the map $e \mapsto P_e$ is continuous. The distribution $P_e$ need not be `objective'; all that matters is that the agent \emph{believes} that output will realise according to the probability distribution $P_e$ if she supplies effort~$e$.

The agent's payoff under a contract $w \in W$ given effort $e \in E$ is
\begin{equation*}
	-C(e) + \sum_{s \in S} u(w(s)) P_e(s) ,
\end{equation*}
where $u : \Pi \to \R$ is a strictly increasing and normalised utility function, whose curvature captures the agent's risk attitude. The agent chooses effort optimally. In general, she may find it strictly optimal to randomise effort. Her payoff under a contract $w \in W$ is therefore
\begin{equation*}
	\sup_{\mu \in \Delta(E)} \int_E \left[
	-C(e) + \sum_{s \in S} u(w(s)) P_e(s)
	\right] \mu(\dd e) ,
\end{equation*}
where $\Delta(E)$ denotes the set of all probability measures on $E$.

%%%%%%%%%%%%%%%%%%%%%%%%%%%%%%%%%%%
\subsection{Data: preferences over contracts}
\label{sec:model:prefs}
%%%%%%%%%%%%%%%%%%%%%%%%%%%%%%%%%%%

The agent's preference over contracts, $\succeq$, is a binary relation on $W$. As usual, we write $w \succ w'$ if and only if $w \succeq w' \nsucceq w$. We interpret the preference $\succeq$ as (potential) data on the agent's choice behaviour: in particular, her choices when presented with any given pair $w,w' \in W$ of contracts. In other words, $\succeq$ is an empirical object.

A binary relation $\succeq$ on $W$ is a \emph{moral-hazard preference} exactly if it may be represented by a moral-hazard model: that is, if and only if there is a non-empty, compact and convex (effort) set $E \subseteq \R^n$ (where $n \in \N$), a grounded and lower semi-continuous (cost) function $C : E \to \R_+$, a continuous (belief) map $e \mapsto P_e$ carrying $E$ into $\Delta(S)$, and a strictly increasing and normalised (utility) function $u : \Pi \to \R$, such that for any contracts $w,w' \in W$, $w \succeq w'$ holds if and only if
\begin{align*}
	&\sup_{\mu \in \Delta(E)} \int_E \left[
	-C(e) + \sum_{s \in S} u(w(s)) P_e(s)
	\right] \mu(\dd e)
	\\
	\geq {}
	&\sup_{\mu \in \Delta(E)} \int_E \left[
	-C(e) + \sum_{s \in S} u(w'(s)) P_e(s)
	\right] \mu(\dd e) .
\end{align*}

We assume that $\succeq$ is the only data that is available. In particular, following the moral-hazard literature, we assume that the agent's effort choices $e \in E$ are unobservable.

%%%%%%%%%%%%%%%%%%%%%%%%%%%%%%%%%%%
%%%%%%%%%%%%%%%%%%%%%%%%%%%%%%%%%%%
\section{A parsimonious moral-hazard model}
\label{sec:parsim}
%%%%%%%%%%%%%%%%%%%%%%%%%%%%%%%%%%%
%%%%%%%%%%%%%%%%%%%%%%%%%%%%%%%%%%%

In the moral-hazard model, effort $e \in E$ is just an index: what the agent actually cares about is the output distribution $P_e \in \Delta(S)$ and the cost $C(e)$. We may therefore recast the moral-hazard model as one in which the agent directly chooses the output distribution $p$ at some cost $c(p)$, as follows.

\begin{lemma}
	\label{lemma:parsimonious}
	A relation $\succeq$ is a moral-hazard preference if and only if there is a grounded, convex and lower semi-continuous function $c : \Delta(S) \to [0,\infty]$ and a strictly increasing and normalised function $u : \Pi \to \R$ such that for any contracts $w,w' \in W$, $w \succeq w'$ holds if and only if
	\begin{equation*}
		\max_{p \in \Delta(S)} \left[
		-c(p) + \sum_{s \in S} u(w(s)) p(s)
		\right]
		\geq \max_{p \in \Delta(S)} \left[
		-c(p) + \sum_{s \in S} u(w'(s)) p(s)
		\right] .
	\end{equation*}
\end{lemma}

We call any pair $(c,u)$ that satisfies the properties in \Cref{lemma:parsimonious} a \emph{parsimonious representation} of the moral-hazard preference $\succeq$. (`Parsimonious' since it has just two parameters, $c$ and $u$.) In the remainder of the paper, we work with parsimonious moral-hazard representations.

The main claim of \Cref{lemma:parsimonious} is that the cost function $c$ may be chosen to be convex. The fact that we may write `$\max$' instead of `$\sup$' follows from the lower semi-continuity of $c$. Note that although the representation `allows' the agent to choose \emph{any} output distribution $p \in \Delta(S)$, constraints can still be captured by setting $c(p)=\infty$ for some distributions $p \in \Delta(S)$.

\begin{proof}
	For the `only if' part, let $\succeq$ be a moral-hazard preference, with effort set $E$, cost function $C$, belief map $e \mapsto P_e$ and utility function $u$. For each output distribution $p \in \Delta(S)$, let $c(p)$ be the least cost at which $p$ may be produced:
	\begin{equation*}
		c(p) \coloneqq
		\inf_{\substack{\mu \in \Delta(E) :\\ \int_E P_e \mu(\dd e) = p}} \int_E C(e) \mu(\dd e) ,
	\end{equation*}
	with the convention that $\inf \varnothing = \infty$. The function $c$ maps $\Delta(S)$ into $[0,\infty]$, is convex and lower semi-continuous by construction,%
		\footnote{Convexity is easily verified. For lower semi-continuity (`lsc'), let $\Delta(\Delta(S))$ be the set of all probabilities on $\Delta(S)$, equipped with the topology of weak convergence. For each $p \in \Delta(S)$, $c(p) = \inf_{\nu \in B(p)} \int_{\Delta(S)} \widetilde{c}(q) \nu(\dd q)$, where $B(p) \coloneqq \{ \nu \in \Delta(\Delta(S)) : \int_{\Delta(S)} q \nu(\dd q) = p \}$ and $\widetilde c(q) \coloneqq \inf_{e \in E : P_e = q} C(e)$ for each $q \in \Delta(S)$, and $\inf \varnothing = \infty$ by convention. The map $\widetilde c : \Delta(S) \to [0,\infty]$ is lsc by Lemma 17.30 in \textcite{AliprantisBorder2006}, since $C$ is lsc and (by the continuity of $e \mapsto P_e$) the correspondence $q \mapsto \{ e \in E : P_e = q \}$ is upper hemi-continuous (`uhc') and compact-valued. Since the map $\nu \mapsto \int_{\Delta(S)} q \nu(\dd q)$ is continuous \parencite[][Proposition 1.1]{Phelps2001}, the correspondence $B : \Delta(S) \rightrightarrows \Delta(\Delta(S))$ is uhc and compact-valued. Thus $c$ is lsc by Lemma 17.30 in \textcite{AliprantisBorder2006}. This argument is adapted from \textcite[][appendix A.4.1]{LipnowskiRavid2020}.}
	and is grounded since $C$ is. And clearly
	\begin{multline*}
		\sup_{\mu \in \Delta(E)} \int_E \left[
		-C(e) + \sum_{s \in S} u(w(s)) P_e(s)
		\right] \mu(\dd e)
		\\
		\begin{aligned}
			&= \sup_{p \in \Delta(S)}
			\sup_{\substack{\mu \in \Delta(E) :\\ \int_E P_e \mu(\dd e) = p}} \left[
			- \int_E C(e) \mu(\dd e)
			+ \sum_{s \in S} u(w(s)) p(s)
			\right]
			\\
			&= \sup_{p \in \Delta(S)}
			\left[
			- c(p) + \sum_{s \in S} u(w(s)) p(s)
			\right] .
		\end{aligned}
	\end{multline*}
	Finally, the supremum is attained since $c$ is lower semi-continuous.

	For the `if' part, let $\succeq$ admit parsimonious representation $(c,u)$. We have
	\begin{multline*}
		\max_{p \in \Delta(S)}
		\left[
		- c(p) + \sum_{s \in S} u(w(s)) p(s)
		\right]
		\\
		= \sup_{\mu \in \Delta(\Delta(S))}
		\int_{\Delta(S)} 
		\left[
		- c(p) 
		+ \sum_{s \in S} u(w(s)) p(s)
		\right]
		\mu( \dd p ) 
	\end{multline*}
	since $c$ is convex and $p \mapsto \sum_{s \in S} u(w(s)) p(s)$ is linear, so that randomisation is not strictly optimal. This is a moral-hazard model, with `effort' $p \in \Delta(S) = E$, cost function $C = c$, and belief map $P_p = p$ for every $p \in \Delta(S)$.
\end{proof}

\begin{remark}
	\label{remark:balazs}
	This paper is concerned with the moral-hazard model's behavioural foundations, rather than with its implications for optimal contracting. Still, it is worth pointing out that the parsimonious formulation of the moral-hazard model is in some ways more tractable than the standard formulation described in \cref{sec:model:MH}, since the cost function $c$ is so well-behaved. \textcite{MirrleesZhou2006,GeorgiadisRavidSzentes2023} use this insight to obtain new results on implementability and optimal contracts.
\end{remark}

%%%%%%%%%%%%%%%%%%%%%%%%%%%%%%%%%%%
%%%%%%%%%%%%%%%%%%%%%%%%%%%%%%%%%%%
\section{Empirical content of the moral-hazard model}
\label{sec:axioms}
%%%%%%%%%%%%%%%%%%%%%%%%%%%%%%%%%%%
%%%%%%%%%%%%%%%%%%%%%%%%%%%%%%%%%%%

In this section, I give a characterisation of the behavioural implications of the moral-hazard model. I first introduce a number of \emph{axioms:} testable properties of choice behaviour which a given binary relation $\succeq$ on $W$ may or may not satisfy. I then show that a binary relation $\succeq$ is a moral-hazard preference if and only if it satisfies these axioms. In other words, these axioms exhaust the empirical implications of the moral-hazard model.

We shall work with `mixtures', defined standardly as follows. For any random remunerations $x,x' \in \Delta(\Pi)$ and any $\alpha \in (0,1)$, we write $\alpha x + (1-\alpha) x'$ for the random remuneration under which the agent's pay is drawn from $x$ with probability $\alpha$, and drawn from $x'$ otherwise.%
	\footnote{Explicitly: $[ \alpha x + (1-\alpha) x' ](A) \coloneqq \alpha x(A) + (1-\alpha) x'(A)$ for each finite $A \subseteq \Pi$.}
For any contracts $w,w' \in W$ and any $\alpha \in (0,1)$, the contract $\alpha w + (1-\alpha) w'$ is given by $[ \alpha w + (1-\alpha) w' ](s) \coloneqq \alpha w(s) + (1-\alpha) w'(s)$ for each output level $s \in S$.

%%%%%%%%%%%%%%%%%%%%%%%%%%%%%%%%%%%
\subsection{Basic axioms}
\label{sec:axioms:tech}
%%%%%%%%%%%%%%%%%%%%%%%%%%%%%%%%%%%

Our first four axioms are standard. The first three express basic `economic' properties, and the fourth is a technical regularity condition.

\begin{axiom}[weak order]
	\label{axiom:weakorder}
	$\succeq$ is complete and transitive.
\end{axiom}

\begin{axiom}[monotonicity]
	\label{axiom:m-mon}
	For any prizes $\pi,\pi' \in \Pi$, $\pi > \pi'$ implies $\pi \succ \pi'$.
\end{axiom}

\begin{axiom}[dominance]
	\label{axiom:s-mon}
	If two contracts $w,w' \in W$ satisfy $w(s) \succeq w'(s)$ for every output level $s \in S$, then $w \succeq w'$.
\end{axiom}

\begin{axiom}[continuity]
	\label{axiom:continuity}
	For any contracts $w,w',w'' \in W$,
	the sets $\{ \alpha \in [0,1] : \alpha w + (1-\alpha) w' \succeq w'' \}$ and $\{ \alpha \in [0,1] : w'' \succeq \alpha w + (1-\alpha) w' \}$ are closed.
\end{axiom}

It is easily verified (using \Cref{lemma:parsimonious}) that any moral-hazard preference satisfies these four axioms.

The first of our two substantial axioms is as follows.

\begin{namedthm}[Quasiconvexity axiom.]
	\label{axiom:quasiconvexity}
	For any contracts $w,w' \in W$ such that $w \succeq w' \succeq w$, it holds that $w \succeq \alpha w + (1-\alpha) w'$ for all $\alpha \in (0,1)$.
\end{namedthm}

\hyperref[axiom:quasiconvexity]{Quasiconvexity} says precisely that the agent is averse to `mixing' contracts: if two contracts are equally good, then the contract obtained by randomly selecting one of them (using an $\alpha$-biased coin toss) is weakly worse. In a parsimonious moral-hazard representation, this axiom is satisfied because tailoring a single `effort level' $p'' \in \Delta(\Pi)$ to the mixture of two contracts $w$ and $w'$ is more difficult than tailoring `effort' separately to each of the two contracts.%
	\footnote{\label{footnote:quasiconvexity_mh}Explicitly, writing $\phi_w(p) \coloneqq -c(p) + \sum_{s \in S} u(w(s)) p(s)$ for any $w \in W$ and $p \in \Delta(S)$, we have $\phi_{\alpha w + (1-\alpha) w'}(p) = \alpha \phi_w(p) + (1-\alpha) \phi_{w'}(p) \leq \alpha \max_{\Delta(S)} \phi_w + (1-\alpha) \max_{\Delta(S)} \phi_{w'}$ for every $p \in \Delta(S)$, so $\max_{\Delta(S)} \phi_{\alpha w + (1-\alpha) w'} \leq \alpha \max_{\Delta(S)} \phi_w + (1-\alpha) \max_{\Delta(S)} \phi_{w'}$.}
(A version of) this observation was first made by \textcite{Dreze1961}.%
	\footnote{His `postulate~P2.1' is \hyperref[axiom:quasiconvexity]{Quasiconvexity}. See also \textcite[ch.~2]{Dreze1987} and \textcite{Machina1984}.}

Mathematically, \hyperref[axiom:quasiconvexity]{Quasiconvexity} is identical to Schmeidler's (\citeyear{Schmeidler1989}) definition of `uncertainty-seeking' (the opposite of `uncertainty-aversion').

\begin{example}
	\label{example:malevolent}
	To see what \hyperref[axiom:quasiconvexity]{Quasiconvexity} rules out, consider a moral-hazard model in which `effort' is chosen (and costs are borne) by malevolent Nature: let $c : \Delta(S) \to [0,\infty]$ be grounded, convex and lower semi-continuous, let $u : \Pi \to \R$ be strictly increasing and normalised, and let $\succeq$ be the preference such that for any contracts $w,w' \in W$, $w \succeq w'$ holds if and only if
	\begin{equation*}
		\min_{p \in \Delta(S)} \left[
		c(p) + \sum_{s \in S} u(w(s)) p(s)
		\right]
		\geq \min_{p \in \Delta(S)} \left[
		c(p) + \sum_{s \in S} u(w'(s)) p(s)
		\right] .
	\end{equation*}
	Then $\succeq$ violates \hyperref[axiom:quasiconvexity]{Quasiconvexity}, except if $c$ is trivial.%
		\footnote{Reasoning similar to that in \cref{footnote:quasiconvexity_mh} shows that $\succeq$ satisfies the \emph{reverse} inequality (`Quasiconcavity'), and that the inequality is strict for some $w,w' \in W$ except if $c$ is trivial.}
	Preferences like $\succeq$, called \emph{variational,} will appear again in \cref{sec:axioms:theorem} below.
\end{example}

%%%%%%%%%%%%%%%%%%%%%%%%%%%%%%%%%%%
\subsection{Independence}
\label{sec:axioms:indep}
%%%%%%%%%%%%%%%%%%%%%%%%%%%%%%%%%%%

Every moral-hazard preference has an expected-utility attitude to \emph{risk:} random remunerations $x \in \Delta(\Pi)$ (that is, \emph{constant} contracts) are evaluated via the expectation $u(x) = \int_\Pi u(\pi) x(\dd \pi)$. By the expected-utility theorem \parencite{VonneumannMorgenstern1947}, a relation $\succeq$ which satisfies Axioms~\ref{axiom:weakorder}--\ref{axiom:continuity} has this `objective-expected-utility' property if it satisfies the following axiom.

\begin{namedthm}[vNM Independence axiom.]
	\label{axiom:vnm-indep}
	For any random remunerations $x,x' \in \Delta(\Pi)$ and any $\alpha \in (0,1)$,
	\begin{equation*}
		x \succeq x'
		\quad \text{implies} \quad
		\alpha x + (1-\alpha) y
		\succeq
		\alpha x' + (1-\alpha) y 
		\quad \text{for any $y \in \Delta(\Pi)$.}		
	\end{equation*}
\end{namedthm}

Moral-hazard preferences do \emph{not} generally have an expected-utility attitude to \emph{(Knightian) uncertainty.} Equivalently (given Axioms~\ref{axiom:weakorder}--\ref{axiom:continuity}), moral-hazard preferences do not generally satisfy the stronger Anscombe--Aumann Independence axiom.%
	\footnote{\label{footnote:aa_indep}The axiom: for any contracts $w,w' \in W$ and any $\alpha \in (0,1)$,
	$w \succeq w'$
	implies
	$\alpha w + (1-\alpha) w''
	\succeq
	\alpha w' + (1-\alpha) w''$ for any contract $w'' \in W$. Equivalence holds by the Anscombe--Aumann (\citeyear{AnscombeAumann1963}) expected-utility theorem.}
They do, however, satisfy the following independence axiom from \textcite{MaccheroniMarinacciRustichini2006}, which is stronger than vNM Independence but weaker than Anscombe--Aumann Independence.

\begin{namedthm}[MMR Independence axiom.]
	\label{axiom:mmr-indep}
	For any contracts $w,w' \in W$ and any $\alpha \in (0,1)$,
	\begin{align*}
		\alpha w + (1-\alpha) y
		&\succeq
		\alpha w' + (1-\alpha) y
		&&\text{for some $y \in \Delta(\Pi)$}
		\\
		\text{implies} \qquad
		\alpha w + (1-\alpha) y'
		&\succeq
		\alpha w' + (1-\alpha) y'
		&&\text{for any $y' \in \Delta(\Pi)$.}		
	\end{align*}
\end{namedthm}

\hyperref[axiom:mmr-indep]{MMR Independence} requires that mixing one random remuneration ($y$) into two contracts ($w$ and $w'$) is much the same as mixing in another ($y'$)---what matters is the probability $(1-\alpha)$ with which the output-independent remuneration is awarded instead of whatever output-dependent remuneration the contracts ($w$ and $w'$) specify. The reason why the probability $\alpha$ matters is that it affects the \emph{slope} of contracts, i.e. how they vary with output $s \in S$; by contrast, whether remuneration is (constant and equal to) $y$ or $y'$ with probability $(1-\alpha)$ does not affect slope. In a parsimonious moral-hazard representation, \hyperref[axiom:mmr-indep]{MMR Independence} is satisfied because replacing $y$ with $y'$ does not affect the optimal choice of `effort' $p \in \Delta(S)$.

\begin{example}
	\label{example:incomefx}
	To understand what \hyperref[axiom:mmr-indep]{MMR Independence} rules out, consider a moral-hazard model with income effects: let $u : \Pi \to \R$ be strictly increasing and normalised, let $V : \Delta(S) \times u(\Delta(\Pi)) \to [-\infty,\infty)$ be quasiconcave with $V(p,\cdot)$ strictly increasing for each $p \in \Delta(S)$ and $\sup_{p \in \Delta(S)} V(p,t) = t$ for each $t \in u(\Delta(\Pi))$, and let $\succeq$ be the preference such that for any contracts $w,w' \in W$, $w \succeq w'$ holds if and only if
	\begin{equation*}
		\max_{p \in \Delta(S)} V
		\left(
		p, \sum_{s \in S} u(w(s)) p(s)
		\right)
		\geq \max_{p \in \Delta(S)} V
		\left(
		p, \sum_{s \in S} u(w'(s)) p(s)
		\right) .
	\end{equation*}
	Then $\succeq$ violates \hyperref[axiom:mmr-indep]{MMR Independence}, except if $V$ has the quasilinear form $V(p,t) = -c(p) + t$ for some (necessarily convex) $c : \Delta(S) \to [0,\infty]$.%
		\footnote{Preferences like $\succeq$, except with $V$ replaced by $(p,t) \mapsto -V(p,-t)$ and `$\max$' replaced by `$\min$', were studied by \textcite{CerreiavioglioEtal2011}.}
\end{example}

%%%%%%%%%%%%%%%%%%%%%%%%%%%%%%%%%%%
\subsection{Axiomatic behavioural characterisation}
\label{sec:axioms:theorem}
%%%%%%%%%%%%%%%%%%%%%%%%%%%%%%%%%%%

The following proposition characterises the behavioural content of the moral-hazard model.

\begin{proposition}
	\label{proposition:axiomatisation}
	For a binary relation $\succeq$ on $W$, the following are equivalent:
	
	\begin{enumerate}
	
		\item \label{item:thm_axioms} $\succeq$ satisfies Axioms~\ref{axiom:weakorder}--\ref{axiom:continuity}, \hyperref[axiom:mmr-indep]{MMR Independence}, and \hyperref[axiom:quasiconvexity]{Quasiconvexity}.

		\item \label{item:thm_repres} $\succeq$ is a moral-hazard preference.

	\end{enumerate}
\end{proposition}

To prove \Cref{proposition:axiomatisation}, we begin with an observation. The \emph{inverse} of a binary relation $\succeq$ on $W$ is the binary relation $\sqsupseteq$ on $W$ such that $w' \sqsupseteq w$ if and only if $w \succeq w'$, for all contracts $w,w' \in W$. Given a binary relation $\sqsupseteq$ on $W$, a grounded, convex and lower semi-continuous function $c : \Delta(S) \to [0,\infty]$, and a non-constant function $v : \Pi \to \R$, we say that $(c,v)$ is a \emph{variational representation} of $\sqsupseteq$ exactly if for any contracts $w,w' \in W$, $w' \sqsupseteq w$ holds if and only if
\begin{equation*}
	\min_{p \in \Delta(S)} \left[
	c(p) + \sum_{s \in S} v(w'(s)) p(s)
	\right]
	\geq \min_{p \in \Delta(S)} \left[
	c(p) + \sum_{s \in S} v(w(s)) p(s)
	\right] .
\end{equation*}

\begin{observation}
	\label{observation:variational}
	For a binary relation $\succeq$ on $W$, a grounded, convex and lower semi-continuous function $c : \Delta(S) \to [0,\infty]$, and a non-constant function $u : \Pi \to \R$, the following are equivalent:
	
	\begin{itemize}
	
		\item $(c,u)$ is a parsimonious moral-hazard representation of $\succeq$.

		\item $(c,-u)$ is a variational representation of the inverse of $\succeq$, and $u$ is strictly increasing and normalised.

	\end{itemize}
\end{observation}

\begin{proof}[Proof of \Cref{proposition:axiomatisation}]
	It is easily verified (using \Cref{lemma:parsimonious}) that \ref{item:thm_repres} implies \ref{item:thm_axioms}. For the converse, suppose that $\succeq$ satisfies \ref{item:thm_axioms}, and let $\sqsupseteq$ be its inverse. By inspection, $\sqsupseteq$ satisfies Axioms~\ref{axiom:weakorder}, \ref{axiom:s-mon} and \ref{axiom:continuity}, \hyperref[axiom:mmr-indep]{MMR Independence}, \emph{Non-degeneracy} (there exist contracts $w,w' \in W$ such that $w \succ w'$) and \emph{Quasiconcavity} (for any contracts $w,w' \in W$ such that $w \sqsupseteq w' \sqsupseteq w$, it holds that $\alpha w + (1-\alpha) w' \sqsupseteq w$ for all $\alpha \in (0,1)$). Thus by Theorem 3 in \textcite{MaccheroniMarinacciRustichini2006}, $\sqsupseteq$ admits a variational representation $(c,v)$. By \Cref{axiom:m-mon}, $-v$ is strictly increasing, so $u(\cdot) \coloneqq [ v(\pi_0) - v(\cdot) ] / [ v(\pi_0) - v(\pi_1) ]$ is strictly increasing and normalised. Hence by \Cref{observation:variational}, $(c,u)$ is a parsimonious moral-hazard representation of $\succeq$, and thus $\succeq$ is a moral-hazard preference by \Cref{lemma:parsimonious}.
\end{proof}

%%%%%%%%%%%%%%%%%%%%%%%%%%%%%%%%%%%
%%%%%%%%%%%%%%%%%%%%%%%%%%%%%%%%%%%
\section{Identification of the moral-hazard model}
\label{sec:ident}
%%%%%%%%%%%%%%%%%%%%%%%%%%%%%%%%%%%
%%%%%%%%%%%%%%%%%%%%%%%%%%%%%%%%%%%

In this section, I show that the parsimonious moral-hazard model is identified: both of its parameters, $c$ and $u$, are recoverable from contract choice data.

\begin{definition}
	\label{definition:unbound}
	A binary relation $\succeq$ on $W$ is \emph{unbounded} if and only if there are random remunerations $x \succ y$ in $\Delta(\Pi)$ such that for any $\alpha \in (0,1)$, we may find random remunerations $z,z' \in \Delta(\Pi)$ that satisfy $y \succ \alpha z + (1-\alpha) x$ and $\alpha z' + (1-\alpha) y \succ x$.
\end{definition}

It is easy to see that a moral-hazard preference $\succeq$ is unbounded if and only if for any parsimonious representation $(c,u)$ of $\succeq$, the function $u$ is unbounded both above and below (equivalently, $u(\Delta(\Pi)) = \R$).

\begin{proposition}[identification]
	\label{proposition:identification}
	Each unbounded moral-hazard preference admits \emph{exactly} one parsimonious representation.
\end{proposition}

\begin{proof}
	By \Cref{lemma:parsimonious}, each moral-hazard preference admits \emph{at least} one parsimonious representation. By Proposition 6 in \textcite{MaccheroniMarinacciRustichini2006} and \Cref{observation:variational}, each unbounded moral-hazard preference admits \emph{at most} one parsimonious representation.%
		\footnote{This argument, and thus \Cref{proposition:identification}, remains valid if unboundedness is weakened to require only that $u(\Delta(\Pi))$ be unbounded \emph{either} above \emph{or} below. The full force of unboundedness will be needed in \cref{sec:conf_optim:optim} below, however.}
\end{proof}

\Cref{proposition:identification} asserts that observing the agent's choices between pairs of contracts suffices to recover the parameters $c$ and $u$. Concretely, $u$ may be recovered in standard fashion from choice between random remunerations (that is, \emph{constant} contracts), whereupon $c$ may be recovered as
\begin{equation*}
	c(p)
	= \sup_{w \in W} \left( - u(x_w) + \sum_{s \in S} u(w(s)) p(s) \right)
	\quad \text{for each $p \in \Delta(S)$,}
\end{equation*}
where $x_w$ denotes the unique random remuneration $x \in \Delta(\Pi)$ that satisfies $x \succeq w \succeq x$.

The standard four-parameter moral-hazard model $(E,C,e \mapsto P_e,u)$ described in \cref{sec:model:MH} is only partially identified off contract choice data $\succeq$. In particular, it is not possible to disentangle the effort cost $C$ from the effort-to-output map $e \mapsto P_e$: the data $\succeq$ can only reveal the minimum cost $c(p)$ of inducing any given output distribution $p \in \Delta(S)$. Richer data could help: for example, data on chosen effort $e \in E$ \parencite[see][]{Karni2006}. Such data is rarely available in the field,%
	\footnote{See \textcite[section~3]{Georgiadis2022} for some notable exceptions.}
but can potentially be obtained in the lab.%
	\footnote{Two common methods are incentivised elicitation and measuring performance on a cognitive task \parencite[see e.g.][]{CharnessGneezyHenderson2018}. A third method is measuring the total time spent on a task \parencite[e.g.][]{AvoyanRibeiroSchotter2024}.}

%%%%%%%%%%%%%%%%%%%%%%%%%%%%%%%%%%%
%%%%%%%%%%%%%%%%%%%%%%%%%%%%%%%%%%%
\section{Confidence and optimism}
\label{sec:conf_optim}
%%%%%%%%%%%%%%%%%%%%%%%%%%%%%%%%%%%
%%%%%%%%%%%%%%%%%%%%%%%%%%%%%%%%%%%

In this section, I propose behavioural definitions `more confident than' and `more optimistic than'. These definitions are couched entirely in terms of observed choice between contracts, making no reference to any objective facts about how the agent's effort \emph{actually} influences output. I show that in the parsimonious moral-hazard model, increased confidence corresponds to `vertical' shifts (precisely: pointwise decreases) of the output-distribution cost function $c$, while increased optimism corresponds to `horizontal' shifts of $c$ in the direction of the first-order stochastically higher output distributions. I also discuss what data is required to distinguish empirically between `more confident than' and `more optimistic than'. I conclude by noting that my notions of \emph{relative} confidence and optimism lead naturally to definitions of what it means for an agent to be (absolutely) \emph{over}confident or optimistic.

%%%%%%%%%%%%%%%%%%%%%%%%%%%%%%%%%%%
\subsection{Confidence}
\label{sec:conf_optim:overconf}
%%%%%%%%%%%%%%%%%%%%%%%%%%%%%%%%%%%

A confident agent is one who believes she can significantly influence the distribution of output. In terms of contract choice behaviour, greater confidence is thus expressed by a greater appetite for non-constant contracts, under which pay depends on realised output. This motivates the following purely behavioural definition of relative confidence.

\begin{definition}
	\label{definition:overconf}
	Let $\succeq$ and $\succeq'$ be binary relations on $W$. $\succeq$ is \emph{more confident than} $\succeq'$ if and only if whenever $w \succeq' \mathrel{(\succ')} x$ for a contract $w \in W$ and a constant contract $x \in \Delta(\Pi)$, we also have $w \succeq \mathrel{(\succ)} x$.
\end{definition}

Note that the definition of relative confidence makes no reference to any objective facts about the agent's \emph{actual} ability to influence the distribution of output. Greater confidence may therefore reflect either an increased subjective confidence on the agent's part about her ability to influence output, or an actual improvement in her ability, or a combination of the two.

Mathematically, \Cref{definition:overconf} is identical to the standard definition of `less uncertainty-averse than' \parencite{Epstein1999,GhirardatoMarinacci2002}.

In the moral-hazard model, greater confidence is equivalent to uniformly lower costs:

\begin{proposition}
	\label{proposition:conf}
	Let $\succeq$ and $\succeq'$ be moral-hazard preferences, with parsimonious representations $(c,u)$ and $(c',u')$. Then $\succeq$ is more confident than $\succeq'$ if and only if $u=u'$ and $c \leq c'$.
\end{proposition}

Equivalently, $(c,u)$ is more confident than $(c',u')$ exactly if $u = u'$ and the cost level sets are nested:
\begin{equation*}
	\left\{ p \in \Delta(S) : c(p) \leq k \right\}
	\supseteq 
	\left\{ p \in \Delta(S) : c'(p) \leq k \right\}
	\quad \text{for every $k \geq 0$.}
\end{equation*}
In other words, any output distribution that is (believed to be) attainable at cost $\leq k$ under $(c',u')$ is also (believed to be) attainable at cost $\leq k$ under $(c,u)$. The fact that $u=u'$ is implied reflects a (conceptually desirable) separation of confidence from risk attitude.

\begin{example}
	\label{example:1d_vertical}
	Let output be binary, $S = \{0,1\}$, so that each distribution $p \in \Delta(S)$ may be identified with a single number, $p \equiv \text{Pr}(\text{output}=1)$. Let $c(p) \coloneqq \alpha (p-\beta)^2$ for each $p \in [0,1]$, where $\alpha \in \R_+$ and $\beta \in [0,1]$ are parameters. As $\alpha$ falls (holding $\beta$ constant), $c$ shifts vertically downwards (see \Cref{fig:confidence}). By \Cref{proposition:conf}, the agent becomes more confident.

	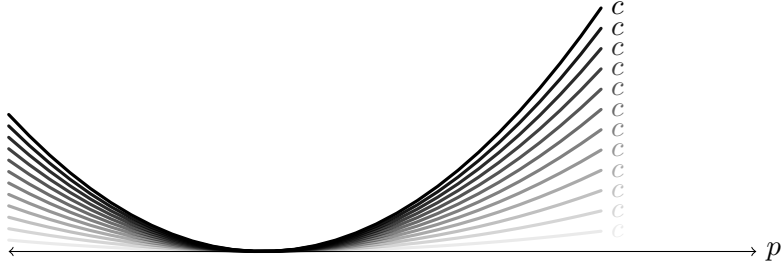
\begin{figure}
		\centering
		\begin{tikzpicture}[xscale=0.75,yscale=0.75,line cap=round]
	\pgfmathsetmacro{\xmax}{10.5};
	\pgfmathsetmacro{\yscale}{0.12};
	\pgfmathsetmacro{\groundmin}{4.5};
	\pgfmathsetmacro{\groundincrement}{0.25};
	\pgfmathsetmacro{\groundmax}{7.5};
	\draw[<->] (0,0) -- ({\xmax+\groundmax-\groundmin-\groundincrement},0);
	\draw ({\xmax+\groundmax-\groundmin-\groundincrement},0) node[anchor=west] {$p$};
	\pgfmathsetmacro{\groundminplusone}{\groundmin+\groundincrement};
	\foreach \ground in {\groundmin,\groundminplusone,...,\groundmax}
		{
		\pgfmathsetmacro{\opac}{1-(\ground-\groundmin)/(\groundmax-\groundmin)}
		\draw[domain=0:\xmax, variable=\x,
			samples=30, very thick, opacity={\opac}]
			plot ( { \x }, { \yscale*\opac*(\x-\groundmin)^2 } );
		\draw[opacity={\opac}] ({\xmax},{ \yscale*\opac*(\xmax-\groundmin)^2 }) node[anchor=west] {$c$};
		}
\end{tikzpicture}
		\caption{As $\alpha$ falls, $p \mapsto \alpha (p-\beta)^2$ shifts vertically downward.}
		\label{fig:confidence}
	\end{figure}
\end{example}

\begin{proof}[Proof of \Cref{proposition:conf}]
	The `if' part is trivial. For the `only if' part, let $\succeq$ and $\succeq'$ be moral-hazard preferences with parsimonious representations $(c,u)$ and $(c',u')$, respectively, and let $\succeq$ be more confident than $\succeq'$. Write $\sqsupseteq$ and $\sqsupseteq'$ for the inverses of $\succeq$ and $\succeq'$, respectively. By \Cref{observation:variational}, $(c,-u)$ and $(c',-u')$ are variational representations of $\sqsupseteq$ and of $\sqsupseteq'$, respectively. And by inspection, for any random remuneration $x \in \Delta(\Pi)$ and any contract $w \in W$, $x \sqsupseteq' w$ implies $x \sqsupseteq w$. Thus $u=u'$ and $c \leq c'$ by Proposition~9 in \textcite{MaccheroniMarinacciRustichini2006}.
\end{proof}

%%%%%%%%%%%%%%%%%%%%%%%%%%%%%%%%%%%
\subsection{Optimism}
\label{sec:conf_optim:optim}
%%%%%%%%%%%%%%%%%%%%%%%%%%%%%%%%%%%

An optimistic agent is one who expects output to be high. To capture this, assume that output levels are ordered: $S = \{s_1,s_2,\dots,s_{\abs*{S}}\}$, where $s_1 < s_2 < \cdots < s_{\abs*{S}}$.

In terms of contract choice behaviour, greater optimism is manifested by a greater propensity to choose `steeper' contracts, which pay relatively more following high output realisations than following low output realisations.

The appropriate formalisation of `steeper' requires correcting for risk attitude, by using units of utility rather than of money. Toward a definition, recall that by the expected-utility theorem \parencite{VonneumannMorgenstern1947}, if a relation $\succeq$ satisfies Axioms~\ref{axiom:weakorder}--\ref{axiom:continuity} and \hyperref[axiom:vnm-indep]{vNM Independence}, then there is exactly one strictly increasing and normalised function $u : \Pi \to \R$ such that for any random remunerations $x,x' \in \Delta(\Pi)$, $x \succeq x'$ holds if and only if $\int_\Pi u(\pi) x(\dd \pi) \geq \int_\Pi u(\pi) x'(\dd \pi)$. Given any such relation $\succeq$ and any two contracts $w,w' \in W$, we say that $w$ is \emph{$\succeq$-steeper than} $w'$ exactly if $s \mapsto u(w(s)) - u(w'(s))$ is increasing. Equivalently, $\frac{1}{2} w(s) + \frac{1}{2} w'(s') \succeq \frac{1}{2} w(s') + \frac{1}{2} w'(s)$ holds for all output levels $s \geq s'$ in $S$. This latter formulation furnishes a general definition that is applicable for \emph{any} relation $\succeq$:

\begin{definition}
	\label{definition:steeper}
	Fix a binary relation $\succeq$ on $W$ and two contracts $w,w' \in W$. We say that $w$ is \emph{$\succeq$-steeper than} $w'$ if and only if $\frac{1}{2} w(s) + \frac{1}{2} w'(s') \succeq \frac{1}{2} w(s') + \frac{1}{2} w'(s)$ holds for all output levels $s \geq s'$ in $S$.
\end{definition}

\begin{definition}
	\label{definition:optim}
	Let $\succeq$ and $\succeq'$ be binary relations on $W$. $\succeq$ is \emph{more optimistic than} $\succeq'$ if and only if whenever $w \succeq' \mathrel{(\succ')} w'$ for two contracts $w,w' \in W$ such that $w$ is $\succeq$-steeper than $w'$, we also have $w \succeq \mathrel{(\succ)} w'$.
\end{definition}

In other words, a more optimistic preference is one that is more disposed to prefer steeper contracts, where steepness is measured in utility units.

We shall characterise relative optimism in the moral-hazard model in terms of `up-shiftedness' of costs, defined as follows. Abbreviate `first-order stochastically dominates' to `FOSD'.

\begin{definition}
	\label{definition:upshift}
	Let $c,c' : \Delta(S) \to [0,\infty]$ be grounded, convex and lower semi-continuous. We say that $c$ is \emph{up-shifted from} $c'$ if and only if for any $p,p' \in \Delta(S)$, there are $q,q' \in \Delta(S)$ such that $p$ FOSD $q'$, $q$ FOSD $p'$, $\frac{1}{2} p + \frac{1}{2} p' = \frac{1}{2} q + \frac{1}{2} q'$, and $c(q) + c'(q') \leq c(p) + c'(p')$.
\end{definition}

Despite its tricky definition, up-shiftedness expresses a straightforward idea: that first-order stochastically higher output distributions are relatively cheaper under $c$ than under $c'$. Indeed, $c$ is up-shifted from $c'$ exactly if under any contract $w \in W$ and for any strictly increasing and normalised $u : \Pi \to \R$, the optimal choice of `effort' $p \in \Delta(S)$ in the parsimonious moral-hazard model $(c,u)$ is first-order stochastically higher than optimal `effort' $p' \in \Delta(S)$ in the parsimonious moral-hazard model $(c',u)$ \parencite[see][section~5.2]{DziewulskiQuah2024}.%
	\footnote{This statement assumes that optimal effort choices are unique. The general statement may be found in \textcite[section~5.2]{DziewulskiQuah2024}.}

\setcounter{example}{2}
\begin{example}[continued]
	\label{example:1d_horizontal}
	Recall that $c(p) = \alpha (p-\beta)^2$ for each $p \in [0,1]$, and note that $p$ FOSD $p'$ if and only if $p \geq p'$. As $\beta$ rises (holding $\alpha$ constant), $c$ shifts horizontally rightwards (see \Cref{fig:optimism}); formally, $c$ up-shifts.%
		\footnote{Fix $\alpha \in \R_+$ and $\beta > \beta'$ in $[0,1]$. For any $p,p' \in [0,1]$, it is easily verified that $q \coloneqq \max\{p,p'\}$ and $q' \coloneqq \min\{p,p'\}$ satisfy the properties in \Cref{definition:upshift}.}

	\begin{figure}
		\centering
		\begin{tikzpicture}[xscale=0.75,yscale=0.75,line cap=round]
	\pgfmathsetmacro{\xmax}{10.5};
	\pgfmathsetmacro{\yscale}{0.12};
	\pgfmathsetmacro{\groundmin}{4.5};
	\pgfmathsetmacro{\groundincrement}{0.25};
	\pgfmathsetmacro{\groundmax}{7.5};
	\draw[<->] (0,0) -- ({\xmax+\groundmax-\groundmin-\groundincrement},0);
	\draw ({\xmax+\groundmax-\groundmin-\groundincrement},0) node[anchor=west] {$p$};
	\pgfmathsetmacro{\groundminplusone}{\groundmin+\groundincrement};
	\foreach \ground in {\groundmin,\groundminplusone,...,\groundmax}
		{
		\pgfmathsetmacro{\opac}{1-(\ground-\groundmin)/(\groundmax-\groundmin)}
		\draw[domain=0:\xmax, variable=\x,
			samples=30, very thick, opacity={\opac}]
			plot ( { \x + \ground - \groundmin }, { \yscale*(\x-\groundmin)^2 } );
		\draw[opacity={\opac}] ({ \xmax + \ground - \groundmin },{ \yscale*(\xmax-\groundmin)^2 }) node[anchor=south] {$c$};
		}
\end{tikzpicture}
		\caption{As $\beta$ rises, $p \mapsto \alpha (p-\beta)^2$ shifts horizontally rightward.}
		\label{fig:optimism}
	\end{figure}
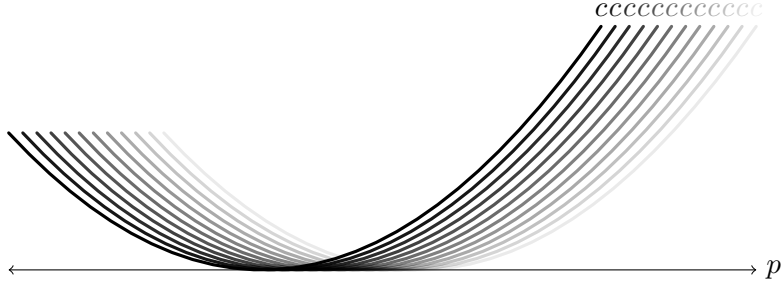
\end{example}

There is an intuitive necessary (though not sufficient) condition for up-shiftedness in terms of the cost level sets
\begin{equation*}
	L_k \coloneqq \{ p \in \Delta(S) : c(p) \leq k \}
	\quad \text{and} \quad
	L_k' \coloneqq \{ p \in \Delta(S) : c'(p) \leq k \} .
\end{equation*}

\begin{observation}
	\label{lemma:shift_wso}
	Let $c,c' : \Delta(S) \to [0,\infty]$ be grounded, convex and lower semi-continuous. If $c$ is up-shifted from $c'$, then for every $k \geq 0$,
	\begin{equation*}
		\begin{aligned}
			&\text{for each $p \in L_k$,}\;
			&&\text{$p$ FOSD $p'$ for some $p' \in L_k'$,}\;
			&&\text{and}
			\\
			&\text{for each $p' \in L_k'$,}\;
			&&\text{$p$ FOSD $p'$ for some $p \in L_k$.}
			&&
		\end{aligned}
	\end{equation*}
\end{observation}

In other words, the set $L_k$ of output distributions that are (believed to be) attainable at cost $\leq k$ under $c$ is `FOSD higher' than the set $L_k'$ of output distributions that are (believed to be) attainable at cost $\leq k$ under $c$.%
	\footnote{This way of comparing sets is known as the \emph{weak set order} \parencite[][section 2.4]{Topkis1998}.}

\begin{proof}
	We prove the first half, omitting the analogous argument for the second half. Fix a $k \geq 0$ and a $p \in L_k$. Since $c'$ is grounded and lower semi-continuous, we may choose a $p' \in \Delta(S)$ such that $c'(p')=0$. By up-shiftedness, there are $q,q' \in \Delta(S)$ such that $p$ FOSD $q'$ and $c'(q') \leq c(q) + c'(q') \leq c(p) + c'(p') \leq k$, so that $q' \in L_k'$.
\end{proof}

We now show that greater optimism is manifested in the moral-hazard model by first-order stochastically higher distributions becoming relatively cheaper, in the sense of up-shiftedness.

\begin{proposition}
	\label{proposition:optim}
	Let $\succeq$ and $\succeq'$ be unbounded moral-hazard preferences, with parsimonious representations $(c,u)$ and $(c',u')$. Then $\succeq$ is more optimistic than $\succeq'$ if and only if $u=u'$ and $c$ is up-shifted from $c'$.
\end{proposition}

The proof, given in \cref{app:pf_optim}, turns on a result due to \textcite{DziewulskiQuah2024}.

%%%%%%%%%%%%%%%%%%%%%%%%%%%%%%%%%%%
\subsection{Distinguishability of confidence and optimism}
\label{sec:conf_optim:distinguish}
%%%%%%%%%%%%%%%%%%%%%%%%%%%%%%%%%%%

Relative confidence and relative optimism are logically independent: by \Cref{example:1d_vertical}, a preference can be more confident than another without being more optimistic, and vice versa.%
	\footnote{Explicitly, there exist binary relations $\succeq$ and $\succeq'$ on $W$ such that $\succeq$ is more confident but not more optimistic than $\succeq'$, and there exist relations $\succeq$ and $\succeq'$ on $W$ such that $\succeq$ is more optimistic but not more confident than $\succeq'$; furthermore, these can be chosen to satisfy Axioms~\ref{axiom:weakorder}--\ref{axiom:continuity}, \hyperref[axiom:mmr-indep]{MMR Independence}, and \hyperref[axiom:quasiconvexity]{Quasiconvexity}.}
In other words, relative confidence is empirically distinguishable from relative confidence.

This conclusion rests in part on the richness of the data $\succeq$, which records the agent's choices between all pairs $w,w' \in W$ of contracts. On coarser data, distinguishability may fail. To explore this point, suppose that data is available only on choices between pairs of contracts belonging to $W^\star \subseteq W$, where $W^\star$ contains the constant contracts (that is, $W^\star \supseteq \Delta(\Pi)$). For binary relations $\succeq$ and $\succeq'$ on $W$, say that $\succeq$ is \emph{more confident on $W^\star$ than} $\succeq'$ if and only if whenever $w \succeq' \mathrel{(\succ')} x$ for a contract $w \in W^\star$ and a constant contract $x \in \Delta(\Pi)$, we also have $w \succeq \mathrel{(\succ)} x$, and say that $\succeq$ is \emph{more optimistic on $W^\star$ than} $\succeq'$ if and only if whenever $w \succeq' \mathrel{(\succ')} w'$ for two contracts $w,w' \in W^\star$ such that $w$ is $\succeq$-steeper than $w'$, we also have $w \succeq \mathrel{(\succ)} w'$.

Whether `more confident on $W^\star$ than' implies `more optimistic on $W^\star$ than', or vice versa, depends on $W^\star \subseteq W$. Consider, for example, the set $W^{\uparrow}$ of all \emph{increasing} contracts, meaning those $w \in W$ such that $w(s)$ first-order stochastically dominates $w(s')$ whenever $s \geq s'$.

\begin{observation}
	\label{observation:increasing}
	For binary relations $\succeq$ and $\succeq'$ on $W$ that satisfy Axioms~\ref{axiom:weakorder}--\ref{axiom:continuity} and \hyperref[axiom:vnm-indep]{vNM Independence}, if $\succeq$ is more optimistic on $W^{\uparrow}$ than $\succeq'$, then $\succeq$ is more confident on $W^{\uparrow}$ than $\succeq'$.%
		\footnote{I thank an anonymous referee for this observation.}
\end{observation}

Thus data on increasing contracts alone is insufficient to distinguish relative confidence from relative optimism.

\begin{proof}
	Fix a contract $w \in W^\star$ and a constant contract $x \in \Delta(\Pi)$ such that $w \succeq' \mathrel{(\succ')} x$; we must show that if $\succeq$ is more optimistic on $W^\uparrow$ than $\succeq'$, then $w \succeq \mathrel{(\succ)} x$. By the expected-utility theorem \parencite{VonneumannMorgenstern1947}, there is exactly one strictly increasing and normalised function $u : \Pi \to \R$ such that for any random remunerations $x',x'' \in \Delta(\Pi)$, $x' \succeq x''$ holds if and only if $\int_\Pi u(\pi) x'(\dd \pi) \geq \int_\Pi u(\pi) x''(\dd \pi)$. Since $u$ is strictly increasing and $w$ is increasing, $w$ is $\succeq$-steeper than $x$. Thus if $\succeq$ is more optimistic on $W^\uparrow$ than $\succeq'$, then $w \succeq \mathrel{(\succ)} x$.
\end{proof}

%%%%%%%%%%%%%%%%%%%%%%%%%%%%%%%%%%%
\subsection{Absolute overconfidence and optimism}
\label{sec:conf_optim:absolute}
%%%%%%%%%%%%%%%%%%%%%%%%%%%%%%%%%%%

A remaining question is what it means for an agent to be \emph{overconfident} or \emph{optimistic,} in an absolute sense (rather than relative to some other agent). I conclude by offering one possible answer to this question, in the spirit of \textcite{Epstein1999,GhirardatoMarinacci2002}.

Whereas our definitions of `more confident than' and `more optimistic than' were purely behavioural, making no reference to objective facts about the agent's ability to influence output, absolute overconfidence and optimism are concerned precisely with discrepancies between such objective facts and the agent's subjective assessment of them. We must therefore fix an `objective' benchmark against which the agent's behaviour may be compared. Although it is natural to fix as our benchmark the first three parameters $(E^\star,C^\star,e \mapsto P_e^\star)$ of a standard four-parameter moral-hazard model, what will actually matter is only the objective output-distribution cost $c^\star : \Delta(S) \to [0,\infty]$.%
	\footnote{The relationship is $c^\star(p) = \inf_{\mu \in \Delta(E^\star)} \{ \int_{E^\star} C^\star(e) \mu(\dd e) : \int_{E^\star} P_e^\star \mu(\dd e) = p \}$ for every $p \in \Delta(S)$, where $\inf \varnothing = \infty$ by convention. This was shown in the proof of \Cref{lemma:parsimonious}.}

\begin{definition}
	\label{definition:abs}
	Let $c^\star : \Delta(S) \to [0,\infty]$ be the objective output-distribution cost, and let $\succeq$ be a binary relation on $W$ that satisfies Axioms~\ref{axiom:weakorder}--\ref{axiom:continuity} and \hyperref[axiom:vnm-indep]{vNM Independence}. Write $u : \Pi \to \R$ for the (unique) normalised utility function of $\succeq$, and let $\succeq^\star$ be the moral-hazard preference whose parsimonious representation is $(c^\star,u)$. We say that $\succeq$ is \emph{overconfident (optimistic)} if and only if it is more confident (more optimistic) than $\succeq^\star$.
\end{definition}

In other words, an overconfident (optimistic) preference is one that is more confident (more optimistic) than is objectively warranted.

\begin{corollary}
	\label{corollary:abs}
	Let $c^\star : \Delta(S) \to [0,\infty]$ be the objective output-distribution cost, and let $\succeq$ be an unbounded moral-hazard preference, with parsimonious representation $(c,u)$. Then $\succeq$ is overconfident if and only if $c \leq c^\star$, and $\succeq$ is optimistic if and only if $c$ is up-shifted from $c^\star$.
\end{corollary}

Analogous definitions and characterisations may be given for (absolute) \emph{underconfidence} and \emph{pessimism.}

%______________________________________________________________________________

%       _                               _ _
%      / \   _ __  _ __   ___ _ __   __| (_) ___ ___  ___
%     / _ \ | '_ \| '_ \ / _ \ '_ \ / _` | |/ __/ _ \/ __|
%    / ___ \| |_) | |_) |  __/ | | | (_| | | (_|  __/\__ \
%   /_/   \_\ .__/| .__/ \___|_| |_|\__,_|_|\___\___||___/
%           |_|   |_|

\begin{appendices}

\crefalias{section}{appsec}
\crefalias{subsection}{appsec}
\crefalias{subsubsection}{appsec}

%%%%%%%%%%%%%%%%%%%%%%
%%%%%%%%%%%%%%%%%%%%%%
\section{Proof of \texorpdfstring{\Cref{proposition:optim}}{Proposition ref{proposition:optim}}}
\label{app:pf_optim}
%%%%%%%%%%%%%%%%%%%%%%
%%%%%%%%%%%%%%%%%%%%%%

The proof relies on an observation and a lemma. The lemma is inspired by, and proved using, a result due to \textcite{DziewulskiQuah2024}.

\begin{observation}
	\label{observation:uw_f}
	If $u : \Delta(\Pi) \to \R$ is onto, then for every function $f : S \to \R$, there is a contract $w \in W$ such that $f(s) = u(w(s))$ for every $s \in S$.
\end{observation}

\begin{lemma}
	\label{lemma:paweljohn}
	Let $c,c' : \Delta(S) \to [0,\infty]$ be grounded, convex, lower semi-continuous. The following are equivalent:
	
	\begin{enumerate}[label=(\alph*)]

		\item $c$ is up-shifted from $c'$.
	
		\item \label{item:paweljohn1}
		For any $f,f' : S \to \R$ such that $f-f'$ is increasing,
		\begin{equation*}
			\max_{p \in \Delta(S)} \left[
			-c'(p) + \sum_{s \in S} f(s) p(s)
			\right]
			\geq \mathrel{(>)}
			\max_{p \in \Delta(S)} \left[
			-c'(p) + \sum_{s \in S} f'(s) p(s)
			\right]
		\end{equation*}
		implies
		\begin{equation*}
			\max_{p \in \Delta(S)} \left[
			-c(p) + \sum_{s \in S} f(s) p(s)
			\right]
			\geq \mathrel{(>)}
			\max_{p \in \Delta(S)} \left[
			-c(p) + \sum_{s \in S} f'(s) p(s)
			\right] .
		\end{equation*}

	\end{enumerate}
\end{lemma}

\begin{proof}[Proof of \Cref{proposition:optim}]
	Suppose that $\succeq$ is more optimistic than $\succeq'$. Then $u=u'$ since for any constant contracts $x,x' \in \Delta(\Pi)$, $x$ is steeper than $x'$ and $x'$ is steeper than $x$. By \Cref{observation:uw_f} (applicable since $\succeq$ and $\succeq'$ are unbounded), property~\ref{item:paweljohn1} holds. Hence by \Cref{lemma:paweljohn}, $c$ is up-shifted from $c'$.

	Suppose that $u=u'$ and that $c$ is up-shifted from $c'$. Then property~\ref{item:paweljohn1} holds by \Cref{lemma:paweljohn}. Since $\succeq$ and $\succeq'$ are unbounded, it follows by \Cref{observation:uw_f} that $\succeq$ is more optimistic than $\succeq'$.
\end{proof}

\begin{proof}[Proof of \Cref{lemma:paweljohn}]
	Property~\ref{item:paweljohn1} is equivalent to:
	
	\begin{enumerate}[label=(\alph*)]

		\setcounter{enumi}{2}

		\item \label{item:paweljohn3}
		For any $f,f' : S \to \R$ such that $f-f'$ is increasing,
		\begin{align*}
			\min_{p \in \Delta(S)} \left[
			c'(p) + \sum_{s \in S} f(s) p(s)
			\right]
			&\geq \mathrel{(>)}
			\min_{p \in \Delta(S)} \left[
			c'(p) + \sum_{s \in S} f'(s) p(s)
			\right]
			\\
			\Longrightarrow\quad
			\min_{p \in \Delta(S)} \left[
			c(p) + \sum_{s \in S} f(s) p(s)
			\right]
			&\geq \mathrel{(>)}
			\min_{p \in \Delta(S)} \left[
			c(p) + \sum_{s \in S} f'(s) p(s)
			\right] .
		\end{align*}

	\end{enumerate}

	\noindent
	We claim that property~\ref{item:paweljohn3} is equivalent to:

	\begin{enumerate}[label=(\alph*),resume]

		\item \label{item:paweljohn4}
		For any $f,f' : S \to \R$ such that $f-f'$ is increasing,
		\begin{align*}
			\min_{p \in \Delta(S)} \left[
			c(p) + \sum_{s \in S} f(s) p(s)
			\right]
			&-
			\min_{p \in \Delta(S)} \left[
			c(p) + \sum_{s \in S} f'(s) p(s)
			\right] 
			\\
			\geq
			\min_{p \in \Delta(S)} \left[
			c'(p) + \sum_{s \in S} f(s) p(s)
			\right]
			&-
			\min_{p \in \Delta(S)} \left[
			c'(p) + \sum_{s \in S} f'(s) p(s)
			\right] .
		\end{align*}

	\end{enumerate}
	It is clear that \ref{item:paweljohn4} implies \ref{item:paweljohn3}. Conversely, if \ref{item:paweljohn4} fails for some $f$ and $f'$ such that $f-f'$ is increasing, then \ref{item:paweljohn3} fails for $f-k$ and $f'$, where
	\begin{equation*}
		k \coloneqq \min_{p \in \Delta(S)} \left[
		c'(p) + \sum_{s \in S} f(s) p(s)
		\right]
		-
		\min_{p \in \Delta(S)} \left[
		c'(p) + \sum_{s \in S} f'(s) p(s)
		\right] .
	\end{equation*}
	
	Finally, property~\ref{item:paweljohn4} holds if and only if $c$ is up-shifted from $c'$, by Proposition~9 in \textcite{DziewulskiQuah2024}.%
		\footnote{The statement of this proposition assumes that $c$ and $c'$ are finite-valued, but that assumption is not required for the proof. (The statement also assumes continuity, but that property is equivalent to lower semi-continuity given the other assumptions.)}
\end{proof}

\end{appendices}

%______________________________________________________________________________

%    ____  _ _     _ _                             _
%   | __ )(_) |__ | (_) ___   __ _ _ __ __ _ _ __ | |__  _   _
%   |  _ \| | '_ \| | |/ _ \ / _` | '__/ _` | '_ \| '_ \| | | |
%   | |_) | | |_) | | | (_) | (_| | | | (_| | |_) | | | | |_| |
%   |____/|_|_.__/|_|_|\___/ \__, |_|  \__,_| .__/|_| |_|\__, |
%                            |___/          |_|          |___/

% \pagebreak
\printbibliography[heading=bibintoc]

%______________________________________________________________________________

\end{document}